\documentclass[11pt]{article}
\usepackage[T1]{fontenc}
\usepackage{microtype}
\usepackage{amsmath,amsthm}
\usepackage{newpxtext,newpxmath}
\usepackage[margin=15pt,font=small,labelfont={bf}]{caption}
\usepackage[margin=1in]{geometry}
\usepackage[english]{babel}
\usepackage[utf8]{inputenc}

\usepackage{graphicx}
\usepackage{algorithm}
\usepackage[noend]{algpseudocode}

\AtBeginDocument{
  }
\usepackage{xcolor}
\usepackage[colorlinks=true,
    citecolor=darkgreen,
    linkcolor=darkred
]{hyperref}
\usepackage{aliascnt}
\definecolor{darkgreen}{rgb}{0,0.5,0.5}
\definecolor{darkred}{rgb}{0.5,0,0}
\newcommand\drop[1]{}
\newcommand*{\newaliastheorem}[3]{
  \newaliascnt{#1}{#2}
  \newtheorem{#1}[#1]{#3}
  \aliascntresetthe{#1}
  \expandafter\newcommand\csname #1autorefname\endcsname{#3}
}
\newtheorem{theorem}{Theorem}
\newaliastheorem{proposition}{theorem}{Proposition}
\newaliastheorem{lemma}{theorem}{Lemma}
\newaliastheorem{corollary}{theorem}{Corollary}
\newaliastheorem{definition}{theorem}{Definition}
\newaliastheorem{conjecture}{theorem}{Conjecture}
\newaliastheorem{example}{theorem}{Example}

\title{Fair Division with Soft Conflicts}
\author{
    Hirotaka Yoneda\thanks{The University of Tokyo, Tokyo, Japan, \texttt{squar37@gmail.com}.\\ 
    Supported by JST ASPIRE JPMJAP2302 and by JST ACT-X JPMJAX25CT.}
    \and 
    Masataka Yoneda\thanks{The University of Tokyo, Tokyo, Japan, \texttt{e869120@gmail.com}.\\ 
    Supported by JST ASPIRE JPMJAP2302.}
}
\date{}

\begin{document}

\maketitle

\begin{abstract}
    We study the fair division of indivisible goods with conflicts between pairs of goods, represented by a graph $G = (V, E)$. We consider ``soft'' conflicts: assigning two adjacent goods to the same agent is allowed, but we seek allocations that are envy-free up to one good (EF1) while keeping the number of such conflict violations small.
    
    We propose a linear-time algorithm for general additive valuations that finds an EF1 allocation with at most $|E|/n + O(|E|^{1-1/(2n-2)})$ violations, for any constant number of agents $n$. The leading term $|E|/n$ matches the worst-case bound on the number of violations. We use a novel approach that combines an algorithm for fair division with cardinality constraints from Biswas \& Barman (2018) and a geometric ``closest points'' argument. For identical additive valuations, we also propose a simple round-robin-based algorithm that finds an EF1 allocation with at most $|E|/n$ violations.
\end{abstract}

\section{Introduction}

Dividing things, such as gifts, tasks, and properties, among several agents is a common real-world problem, with applications ranging from fair course allocation \cite{Bud17} to fair inheritance division \cite{GP15}.

When dividing goods, it is important to consider \emph{envy}: agent X envies agent Y if X prefers Y's share to their own. Ideally, the allocation should be \emph{envy-free} for all agents. However, achieving envy-freeness may be impossible when goods are indivisible; for example, when dividing a single good between two agents. Therefore, a relaxed fairness criterion, \emph{envy-free up to one good (EF1)}, introduced by Budish (2011) \cite{Bud11}, is widely used in fair division research. Under this criterion, envy is allowed but only up to the removal of at most one good from the envied agent's bundle. EF1 allocations are known to exist under the natural assumption that valuations are additive \cite{LMM+04,CKM+19}.

However, in many practical settings, there are constraints on how division is performed \cite{Suk21}. We highlight the following example as a motivating application.

\begin{description}
    \item[\textbf{Class division in schools.}] In a school, $m$ students will be divided into $n$ classes, each taught by a different teacher. The basic setup is that each teacher has preferences over the students, and the students must be divided so that no teacher envies another teacher. However, we also need to consider the students' relationships. In particular, two students in a bad relationship are not desired to be in the same class.
\end{description}

If we translate this situation into the setting of fair division of indivisible goods, we have \emph{conflicts}, i.e., pairs of goods that are \emph{not desired} to be assigned to the same agent. Conflicts are represented by an undirected graph $G = (V, E)$, where the vertices represent goods and the edges represent conflicts.

Previous studies on fair division with conflict constraints \cite{CKM+23,HH22,KEG+24,IMY25} consider settings where conflicting goods \emph{cannot} be assigned to the same agent; these are ``hard'' conflicts. For example, consider the fair division of jobs in a team, where each job is specified by a working-time interval. In this setting, it is physically impossible to assign two jobs with overlapping intervals to the same agent. Consequently, the team may need to give up some jobs (e.g., when the graph is not $n$-colorable, where $n$ is the number of team members).

However, in the context of class division, allocating all students is paramount, because otherwise teachers would need to expel some students from school. Therefore, teachers must accept some conflicting students in the same class. In this paper, we study fair division with ``soft'' conflict constraints, which allow some conflicting goods to be assigned to the same agent (``conflict violations'') while requiring that the number of violations be kept small.

\paragraph{The baseline.}
We consider the natural benchmark of fair division with soft conflicts to be the task of finding an EF1 allocation with at most $|E|/n$ violations. We provide three reasons below:
\begin{enumerate}
    \item If we assign each good independently and uniformly at random, the expected number of violations is $|E|/n$.
    \item For every number of agents $n$, there exists an instance with identical additive valuations in which every EF1 allocation has $|E|/n$ or more violations (\autoref{prop:balanced1}). Therefore, the bound $|E|/n$ is tight if achieved.
    \item Even without fairness constraints, it is hard to minimize or even to approximate the number of violations. The two-agent case is equivalent to the Min UnCut problem; an $(\frac{11}{8}-\varepsilon)$-approximation is NP-hard for any $\varepsilon > 0$ \cite{HHM+17,MM17}, and assuming the Unique Games Conjecture, even a constant-factor approximation in polynomial time is impossible \cite{KKM+07}.
\end{enumerate}

\paragraph{Our contributions.}
We study the fair division of indivisible goods under soft conflict constraints. Throughout the paper, we assume that the valuations are \emph{additive}. Our main contributions are as follows (for $n$ agents and $m$ goods):
\begin{enumerate}
    \item \textbf{Identical additive valuations.} For identical additive valuations, we propose the \emph{cyclic shift round-robin} algorithm to compute an EF1 allocation with at most $|E|/n$ violations in $O(m \log m + |E|)$ time. We also extend this result to the two-agent case with general additive valuations, by using the cut-and-choose protocol \cite{BT96}.
    \item \textbf{General additive valuations.} For general additive valuations, we propose an algorithm to compute an EF1 allocation with at most $|E|/n + O(|E|^{1-1/(2n-2)})$ violations in $O(m + |E|)$ time, for any constant $n \geq 3$. This result almost reaches the $|E|/n$ baseline, with a $1+o(1)$ factor.

    To outline the algorithm, we use the algorithm of Biswas \& Barman (2018) \cite{BB18} as a starting point, which finds an EF1 allocation in a constrained setting where goods are divided into $m/n$ categories of size $n$, and the goods in each category must be assigned to distinct agents. We point out that their algorithm also works in an ``online algorithm'' setting, where the $(i+1)$-th category is given after the goods in the $i$-th category have been assigned.\footnote{We consider the setting in which, after a good is assigned, we cannot change the bundle to which the good belongs. However, swapping bundles between agents is allowed.} We consider this process as a game between a player and an adversary that repeats rounds in which the player chooses $n$ remaining goods and the adversary assigns them. Importantly, the player can always choose $n$ goods so that the increase in violations is not too large, regardless of the adversary's choice. We prove this fact based on the following ``closest points'' argument: for $kn$ points inside a $d$-dimensional hypercube with side length $\Delta$, there exists a set of $n$ points such that the pairwise distances are $O(\Delta \cdot k^{-1/d})$ (\autoref{lem:hypercube}).
\end{enumerate}

Moreover, in both of our algorithms, the computed allocation is \emph{balanced}; the number of goods each agent takes differs by at most $1$. This property is convenient in some practical settings, such as the class-division example, because we can balance the number of students between classes.

\paragraph{Positioning of our research and related works.}
We explain three important background directions, together with related works, to clarify the positioning of our research.
\begin{enumerate}
    \item The line of research on fair division under conflict constraints, initiated by Chiarelli et al. (2023) \cite{CKM+23}, focuses on ``hard'' conflicts. Hummel and Hetland (2022) \cite{HH22} studied whether EF1 allocations exist when $\Delta(G) < n$, where $\Delta(G)$ is the maximum degree of $G$. Then, to extend the problem to general graphs, Kumar et al. (2024) \cite{KEG+24} studied whether \emph{maximal} EF1 allocations exist, which allow some goods to remain unassigned if these goods conflict with every agent. Igarashi, Manurangsi, Yoneda (2025) \cite{IMY25} showed that maximal EF1 allocations always exist for two agents, but may not exist for three or more agents. To the best of our knowledge, this paper is the first study of fair division with ``soft'' conflicts.
    \item The problem of achieving $|E|/n$ violations is fundamentally more difficult than fair division under cardinality constraints, studied by Biswas \& Barman (2018) \cite{BB18}. In this setting, goods are divided into categories, and the goal is to find an EF1 allocation such that each category is balanced. They showed a positive result for general additive valuations. Returning to fair division under conflict constraints, we observe that when the conflict graph is composed of components of complete graphs, each component must have an ``almost'' balanced allocation to achieve $|E|/n$ violations (\autoref{prop:balanced2}).
    \item If we ignore fairness, the problem becomes maximizing the number of non-violating edges, which is equivalent to Max-$k$-Cut, or Max-Cut for the two-agent case. Max-Cut is one of Karp's 21 NP-complete problems \cite{Kar09}. However, when we choose a uniformly random cut, the size of the cut is at least $|E|/2$ \cite{Erd69}, which matches our baseline. A landmark result by Goemans \& Williamson (1995) \cite{GW95} showed a polynomial-time $0.878$-approximation algorithm for Max-Cut using semidefinite programming.
\end{enumerate}

\section{Preliminaries}
\label{sec:prelim}

\paragraph{Fair division of indivisible goods.}
We denote the set of \emph{agents} by $N = \{1, \dots, n\}$, and the set of \emph{goods} by $M = \{1, \dots, m\}$. Each agent $i \in N$ has a \emph{valuation function} $v_i: 2^M \to \mathbb{R}_{\geq 0}$, where $v_i(S)$ indicates how agent $i$ values the set of goods $S \subseteq M$. Throughout this paper, we assume that the \textbf{valuations are additive}, i.e., $v_i(S) = \sum_{g \in S} v_i(\{g\})$ holds for all $i \in N$ and $S \subseteq M$. When all agents have the same valuation function, we say that the valuations are \emph{identical} and denote each agent's valuation by $v(S)$ for each $S \subseteq M$.

An \emph{allocation} is represented as $\mathcal{A} = (A_1, \dots, A_n)$, where $A_i \subseteq M$ is the set of goods assigned to agent $i$. Each of $A_1, \dots, A_n$ is called a \emph{bundle}. The bundles must be disjoint, i.e., $A_i \cap A_j = \emptyset$ for all $i, j \in N \ (i \neq j)$. The allocation is \emph{complete} if $A_1 \cup \dots \cup A_n = M$, and otherwise it is \emph{partial}. The output allocation must be complete. We seek an \emph{EF1 (envy-free up to one good)} allocation, which is defined as follows:

\begin{definition}[EF1 \cite{Bud11}]
    The allocation $\mathcal{A} = (A_1, \dots, A_n)$ is EF1 if, for all $i, j \in N \ (i \neq j)$, the following holds: either $A_j = \emptyset$, or $v_i(A_i) \geq v_i(A_j \setminus \{g\})$ for some $g \in A_j$.
\end{definition}

\paragraph{Conflicts.}
We consider the \emph{conflict graph} $G = (M, E)$ on goods, where, for each $\{u, v\} \in E$, it is undesirable to assign goods $u$ and $v$ to the same agent. When $u$ and $v$ are assigned to the same agent, i.e., $\{u, v\} \subseteq A_i$ for some $i \in N$, we say that the edge $\{u, v\}$ is \emph{violated}. In general, an instance of the problem is given by $(N, M, \mathcal{V}, G)$, where $\mathcal{V} = (v_1, \dots, v_n)$. Our goal is to find an EF1 allocation while keeping the number of violations small.

\paragraph{Balanced allocations.} The allocation $\mathcal{A} = (A_1, \dots, A_n)$ is \emph{balanced} if the difference in the number of goods between agents is at most $1$, i.e., $\bigl||A_i|-|A_j|\bigr| \leq 1$ for all $i, j \in N$.

\paragraph{Dummy goods.} The \emph{dummy goods}, defined as goods with valuation $0$ for all agents, have no effect on EF1. In this paper, we add some dummy goods, which do not conflict with any other goods, to make the number of goods $m$ divisible by $n$. After computing an EF1 allocation in the presence of dummy goods, we simply remove them when outputting the final allocation.

\paragraph{Notation.} In a graph $G = (V, E)$, we denote the degree of a vertex $v \in V$ by $d(v)$, and for $X \subseteq V$, we denote the number of vertices in $X$ that are adjacent to $v$ by $d_X(v)$. Also, for $X \subseteq V$, we denote the set of edges inside $X$ by $E(X)$, and for $X, Y \subseteq V \ (X \cap Y = \emptyset)$, we denote the set of edges between $X$ and $Y$ by $E(X, Y)$. For $a, b \in \mathbb{N}$, we denote the remainder of $a$ divided by $b$ by ``$a \bmod b$''. For $k \in \mathbb{N}$, we denote the set of permutations of $(1, \dots, k)$ by $S_k$.

\section{Basic Facts}

In this section, we present two propositions to clarify the positioning of this problem. First, we show that the bound $|E|/n$ is the best possible.

\begin{proposition}
    \label{prop:balanced1}
    For every $n$, there exists an instance with identical additive valuations for which every EF1 allocation has at least $|E|/n$ violations.
\end{proposition}

\begin{proof}
    We consider an instance with $m = n+1$ goods and the star graph as a conflict graph, where good $i$ and good $n+1$ conflict for each $i = 1, \dots, n$. The valuations are identical and additive, with $v(\{1\}) = \dots = v(\{n\}) = 1$ and $v(\{n+1\}) = 0$. Then, for an allocation to be EF1, each agent must take exactly one valued good (goods $1, \dots, n$); otherwise, an agent takes at least two valued goods and another agent takes no valued good, which is not EF1. Hence, regardless of which agent takes good $n+1$, one violation occurs. Note that $|E|/n = 1$ in this case. See \autoref{fig:counterexample} for an example.
\end{proof}

\begin{figure}[htbp]
	\centering
	\includegraphics[width=0.6\textwidth]{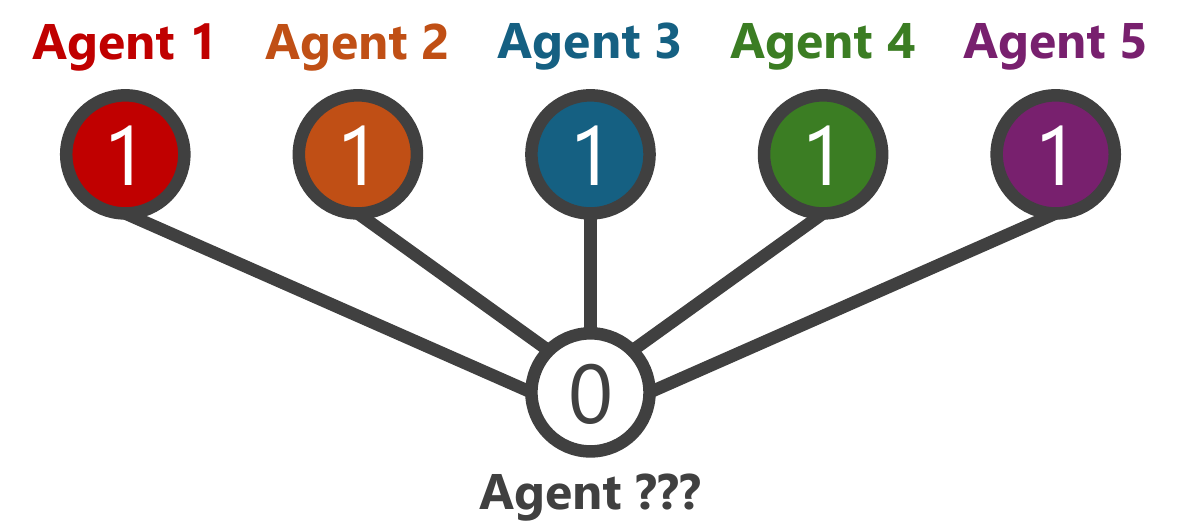}
	\caption{An instance with $n = 5$ agents such that any EF1 allocation has at least $|E|/n$ violations. The number in each vertex denotes the valuation of the good.}
	\label{fig:counterexample}
\end{figure}

Next, we show that this problem is ``almost'' a generalization of fair division under cardinality constraints \cite{BB18}.

\begin{proposition}
    \label{prop:balanced2}
    Consider an instance in which the conflict graph has $k$ connected components, where the $i$-th component is a complete graph with $c_i$ vertices ($c_1 + \dots + c_k = m$). If an allocation has at most $|E|/n$ violations, then for each $i$, every agent must take between $c_i/n - \sqrt{m}$ and $c_i/n + \sqrt{m}$ goods from the $i$-th component.
\end{proposition}

\begin{proof}
    Let $x_{i,j}$ be the number of goods that agent $j$ takes from the $i$-th component. Then:
    \begin{align*}
        \sum_{i=1}^k \sum_{j=1}^n \left(x_{i,j} - \frac{c_i}{n}\right)^2 & = \sum_{i=1}^k \sum_{j=1}^n x_{i,j}^2 + \sum_{i=1}^k \left(\frac{c_i^2}{n^2} - \frac{2 c_i}{n} \sum_{j=1}^n x_{i,j}\right) \\
        & = \sum_{i=1}^k \sum_{j=1}^n x_{i,j}^2 - \sum_{i=1}^k \frac{c_i^2}{n} \\
        & = \sum_{i=1}^k \sum_{j=1}^n \left(2 \cdot \binom{x_{i,j}}{2} + x_{i,j}\right) - \sum_{i=1}^k \left(\frac{2}{n} \cdot \binom{c_i}{2} + \frac{1}{n} \cdot c_i\right) \\
        & = 2 \cdot \left((\text{number of violations}) - \frac{|E|}{n}\right) + \frac{n-1}{n} \cdot m
    \end{align*}
    The first equation holds because $(x_{i,j} - c_i/n)^2 = x_{i,j}^2 + c_i^2/n^2 - (2 c_i / n) \cdot x_{i,j}$. The second equation holds because, for each $i$, $\sum_{j=1}^n x_{i,j} = c_i$. The fourth equation holds because $\sum_{i=1}^k \sum_{j=1}^n x_{i,j} = \sum_{i=1}^k c_i = m$, the number of violations is $\sum_{i=1}^k \sum_{j=1}^n \binom{x_{i,j}}{2}$, and the number of edges is $\sum_{i=1}^k \binom{c_i}{2}$. Therefore, if the number of violations is at most $|E|/n$, $\sum_{i=1}^k \sum_{j=1}^n (x_{i,j} - c_i/n)^2 \leq \frac{n-1}{n} \cdot m$ must hold, so $x_{i,j}$ must be between $c_i/n - \sqrt{m}$ and $c_i/n + \sqrt{m}$.
\end{proof}

\section{Identical Additive Valuations}

In this section, we present an algorithm for identical additive valuations, called the \emph{cyclic shift round-robin} algorithm, which finds an EF1 allocation with at most $|E|/n$ violations.

\subsection{The Round-Robin Method}

We start with the well-known \emph{round-robin algorithm}. The algorithm consists of $\lceil m/n \rceil$ rounds. In each round, each agent (in any order) takes their most preferred remaining good. The round-robin algorithm is known to generate an EF1 allocation for additive valuations \cite{CKM+19}.

We consider identical valuations, so we assume that the goods are sorted in decreasing order of valuations, i.e., $v(\{1\}) \geq \dots \geq v(\{m\})$. We also assume that $m$ is divisible by $n$ (see ``dummy goods'' in \autoref{sec:prelim}). In this case, the following definition characterizes the allocations obtained by the round-robin algorithm.

\begin{definition}
    An allocation $\mathcal{A} = (A_1, \dots, A_n)$ is called a \emph{round-robin allocation} if, for all $i = 1, \dots, m/n$, each agent has exactly one good from $\{(i-1)n+1, \dots, in\}$.
\end{definition}

\begin{lemma}[\cite{CKM+19}]
    \label{lem:round-robin}
    Round-robin allocations are EF1.
\end{lemma}

\begin{proof}
    Let $\mathcal{A} = (A_1, \dots, A_n)$ be a round-robin allocation. To prove that $\mathcal{A}$ is EF1, we show that for all $i, j \in N \ (i \neq j)$, there exists a good $g \in A_j$ such that $v(A_i) \geq v(A_j \setminus \{g\})$.

    Let $A_i = \{x_1, \dots, x_{m/n}\} \ (x_1 < \dots < x_{m/n})$ and $A_j = \{y_1, \dots, y_{m/n}\} \ (y_1 < \dots < y_{m/n})$. Since $\mathcal{A}$ is a round-robin allocation, $x_t < y_{t+1}$ holds. The goods are sorted by valuations, so we have $v(\{x_t\}) \geq v(\{y_{t+1}\})$. Summing over $t = 1, \dots, m/n-1$, we have $v(A_i \setminus \{x_{m/n}\}) \geq v(A_j \setminus \{y_1\})$, which shows $v(A_i) \geq v(A_j \setminus \{g\})$ with $g = y_1$.
\end{proof}

\subsection{Algorithm}

Now, the goal is to create a round-robin allocation with at most $|E|/n$ violations.

We create an allocation in a round-robin fashion, consisting of $m/n$ rounds, where in the $i$-th round, we assign goods in $M_i = \{(i-1)n+1, \dots, in\}$, one good to each of the $n$ agents. The key observation is that when we assign goods uniformly at random, the expected increase in the number of violations is $|E(M_1 \cup \dots \cup M_{i-1}, M_i)|/n$. Therefore, if we optimize the assignment of goods in $M_i$, the violation increase will be at most $|E(M_1 \cup \dots \cup M_{i-1}, M_i)|/n$.

Instead of searching all $n!$ possible assignments, it suffices to consider only $n$ of them: we fix $s \in \{1, \dots, n\}$ and assign goods $(i-1)n+s, \dots, (i-1)n+n, (i-1)n+1, \dots, (i-1)n+s-1$ to agents $1, \dots, n$, respectively. Then, we have the following:

\begin{lemma}
    \label{lem:cyclic1}
    For some $s$, the increase in the number of violations is at most $|E(M_1 \cup \dots \cup M_{i-1}, M_i)|/n$.
\end{lemma}

\begin{proof}
    Each edge $\{x, y\} \in E$ with $x \in M_1 \cup \dots \cup M_{i-1}, y \in M_i$ is violated for exactly one $s$. This is because good $y$ is assigned to a specific agent (namely, the owner of $x$) for exactly one $s$. In addition, there are no violations inside $M_i$, because the goods in $M_i$ are assigned to distinct agents. Therefore, the total increase in the number of violations for $s = 1, \dots, n$ is exactly $|E(M_1 \cup \dots \cup M_{i-1}, M_i)|$, and \autoref{lem:cyclic1} follows.
\end{proof}

We repeat this ``optimized assignment'' for $m/n$ rounds, i.e., for $i = 1, \dots, m/n$. This algorithm is presented as \textsc{CyclicShiftRR} (\autoref{alg:cyclic_round_robin}).

\begin{theorem}
    \textsc{CyclicShiftRR} computes a balanced EF1 allocation with at most $|E|/n$ violations, in $O(m \log m + |E|)$ time.
\end{theorem}

\begin{proof}
    First, the resulting allocation is EF1 because it is a round-robin allocation, where \autoref{lem:round-robin} can be used. The allocation is also balanced; even if dummy items are present, they are assigned in the last round, so every agent receives zero or one dummy item.
    
    Next, we bound the number of violations. Due to \autoref{lem:cyclic1}, there are at most $\sum_{i=2}^{m/n} |E(M_1 \cup \dots \cup M_{i-1}, M_i)|/n$ violations. Since each edge appears in $E(M_1 \cup \dots \cup M_{i-1}, M_i)$ for at most one $i$, the number of violations is at most $|E|/n$.

    Finally, we discuss the time complexity. For each round, we need to compute, for each $s$, the increase in the number of violations for $\mathcal{A}^{(s)}$ compared to the previous $\mathcal{A}$. We can calculate this value by scanning each edge incident to the vertices in $M_i$ and determining for which $\mathcal{A}^{(s)}$ this edge is violated. This computation can be performed in $O(1)$ time for each edge and an additional $O(n)$ time for each round. Therefore, the algorithm runs in $O(m \log m + (m/n) \cdot n + |E|)$ time, considering that we need to sort the goods in decreasing order of valuations.
\end{proof}

\begin{algorithm}
	\caption{$\textsc{CyclicShiftRR}(N, M, \mathcal{V}, G)$}
	\label{alg:cyclic_round_robin}
	\begin{algorithmic}[1]
        \Require $v(\{1\}) \geq \dots \geq v(\{m\})$ and $m$ is divisible by $n$
        \State $\mathcal{A} \gets (\emptyset, \dots, \emptyset)$
        \For{$i = 1, \dots, m/n$}
            \For{$s = 1, \dots, n$}
                \State $c \gets ((i-1)n+s, \dots, (i-1)n + n, (i-1)n + 1, \dots, (i-1)n + s - 1)$
                \State $\mathcal{A}^{(s)} \gets (A_1 \cup \{c_1\}, \dots, A_n \cup \{c_n\})$
            \EndFor
            \State $\mathcal{A} \gets$ the allocation with the minimum number of violations among $\mathcal{A}^{(1)}, \dots, \mathcal{A}^{(n)}$
        \EndFor
		\State \textbf{return} $\mathcal{A}$
	\end{algorithmic}
\end{algorithm}

\subsection{Two Agents: General Additive Case}

We consider the two-agent case specifically. In this case, we can use the well-known \emph{cut-and-choose protocol} \cite{BT96} to extend the result to non-identical valuations.

\begin{theorem}
    \label{thm:two-agents}
    For two agents and general additive valuations, we can compute an EF1 allocation with at most $|E|/2$ violations in $O(m \log m + |E|)$ time.
\end{theorem}

\begin{proof}
    First, using \textsc{CyclicShiftRR}, we compute an allocation with at most $|E|/2$ violations that is EF1 in a hypothetical scenario where the valuations of both agents are $v_1$. Let $(A_1, A_2)$ be the computed allocation. Then, agent 2 takes their preferred bundle (with respect to $v_2$), and agent 1 takes the remaining bundle. The resulting allocation is EF1 and has at most $|E|/2$ violations, as allocations $(A_1, A_2)$ and $(A_2, A_1)$ have the same violations. The time complexity is the same as \textsc{CyclicShiftRR}.
\end{proof}

\section{General Additive Valuations}

In this section, we consider general additive valuations. We show that, for any constant $n \geq 3$ agents, we can achieve EF1 with at most $(1+o(1)) |E|/n$ violations.

\subsection{The Algorithm of Biswas \& Barman (2018)}

A key tool we use is the algorithm of Biswas \& Barman (2018) \cite{BB18} for achieving EF1 under cardinality constraints. We consider a special case in which the goods are divided into $m/n$ categories, each of size $n$. The goods in each category must be assigned to distinct agents. It can be shown that these are the hardest cardinality constraints.

\paragraph{Envy-cycle elimination.}
Their algorithm is based on the \emph{envy-cycle elimination} technique proposed by Lipton et al. (2004) \cite{LMM+04}.

\begin{definition}
    For an allocation $\mathcal{A} = (A_1, \dots, A_n)$, we define the \emph{envy graph} $\mathcal{G}(\mathcal{A}) = (N, E)$, which is a directed graph on agents with the following edges:
    \begin{equation*}
        E = \{(i, j) \mid i, j \in N, v_i(A_i) < v_i(A_j)\}
    \end{equation*}
    The graph represents ``who envies whom.''
\end{definition}

\begin{lemma}[\cite{LMM+04}]
    \label{lem:envy-cycle}
    Let $\mathcal{A} = (A_1, \dots, A_n)$ be an EF1 allocation. Then, there exists a permutation $\sigma \in S_n$ such that $\mathcal{A}_{\sigma} = (A_{\sigma(1)}, \dots, A_{\sigma(n)})$ is EF1 and $\mathcal{G}(\mathcal{A}_{\sigma})$ is a directed acyclic graph.
\end{lemma}

\paragraph{Overview of their algorithm.}
The Biswas \& Barman's algorithm consists of $m/n$ rounds; in the $i$-th round, we assign the goods of the $i$-th category. Before each round, we rearrange the bundles to make the envy graph acyclic (\autoref{lem:envy-cycle}).\footnote{The statement of \autoref{lem:envy-cycle} is different from Biswas \& Barman (2018) \cite{BB18} (Lemma 1). The lemma still holds because in the envy-cycle elimination algorithm, only the rearrangement of bundles happens before assigning the next good. Also see the discussion in \cite{AAB+23} (Section 3.1).} Then, each agent, in topological order of the envy graph, chooses their most preferred remaining good in this category. Using this topological order is natural, since agents with higher envy are assigned better goods.

Their algorithm is presented in \autoref{alg:biswas_barman}, where $\mathcal{C} = (C_1, \dots, C_{m/n})$ is the division of the goods into categories such that $|C_i| = n$ for each $i$.

\begin{algorithm}
	\caption{$\textsc{BiswasBarman}(N, M, \mathcal{V}, \mathcal{C})$}
	\label{alg:biswas_barman}
	\begin{algorithmic}[1]
        \State $\mathcal{A} \gets (\emptyset, \dots, \emptyset)$ \Comment{$\mathcal{A} = (A_1, \dots, A_n)$: current allocation}
        \For{$i = 1, \dots, m/n$}
            \State Rearrange the bundles $A_1, \dots, A_n$ using \autoref{lem:envy-cycle}
            \State $\pi \gets$ a topological order of $\mathcal{G}(\mathcal{A})$
            \State $R \gets C_i$ \Comment{$R$: the set of remaining goods in this category}
            \For{$j = 1, \dots, n$}
                \State $g \gets \operatorname{argmax}_{r \in R} v_{\pi(j)}(\{r\})$
                \State $A_{\pi(j)} \gets A_{\pi(j)} \cup \{g\}$
                \State $R \gets R \setminus \{g\}$
            \EndFor
        \EndFor
		\State \textbf{return} $\mathcal{A}$
	\end{algorithmic}
\end{algorithm}

\subsection{The Strategy: Gameplay}

A key advantage of the Biswas \& Barman's algorithm is that it processes the categories in an \emph{online} manner, without considering information about later categories. Therefore, we can process the categories in any order. This idea also appears in other fair division studies, such as \cite{DFS23}.

More importantly, we can decide the $i$-th category after processing the $(i-1)$-th category; that is, the division itself can be determined online.

\paragraph{The gameplay setting.} We reinterpret the procedure of the Biswas \& Barman's algorithm as a game between a player and an adversary. The game consists of $m/n$ rounds and constructs an allocation. Each round assigns $n$ goods in the following way:
\begin{enumerate}
    \item The player chooses $n$ remaining goods.
    \item The adversary rearranges the bundles of the current allocation and then distributes $n$ goods that the player chose, one good to each of the $n$ agents.
\end{enumerate}
The adversary behaves according to the Biswas \& Barman's algorithm (lines 3--9 of \autoref{alg:biswas_barman}) to ensure that the final allocation is EF1. However, the player cannot control which of the $n!$ ways is chosen in Step 2. We therefore refer to it as the \emph{adversary}; it is convenient to assume the worst-case scenario in which the adversary acts to maximize the final number of violations.

\subsection{Key Idea: Closest Points}

In Step 1, how do we choose $n$ remaining goods so that the increase in violations is not too large, regardless of the adversary's choice? We answer this question using a geometric ``closest points'' argument.

\paragraph{Profile vector.}
We define a \emph{profile vector} for each remaining good, which indicates how the number of violations increases when the good is assigned. Let $\mathcal{A} = (A_1, \dots, A_n)$ be the current allocation. For each $g \in M \setminus (A_1 \cup \dots \cup A_n)$, its profile vector is an $(n-1)$-dimensional vector given by:
\begin{equation*}
    p_g = (d_{A_2}(g) - d_{A_1}(g), \dots, d_{A_n}(g) - d_{A_1}(g))
\end{equation*}
Note that $d_{A_i}(g)$ equals the violation increase when $g$ is put into the bundle $A_i$.

\paragraph{Closest points.} The following lemma illustrates the benefit of choosing $n$ goods whose profile vectors are similar (see also \autoref{fig:profile-vector}).

\begin{lemma}
    \label{lem:matching}
    Let $g_1, \dots, g_n$ be the goods chosen by the player in a round, and let $\lambda \geq 0$. Suppose $\max_{i \in \{1, \dots, n\}} (p_{g_i})_j - \min_{i \in \{1, \dots, n\}} (p_{g_i})_j \leq \lambda$ for all $j = 1, \dots, n-1$. Then, the violation increase in this round is at most $|E(A_1 \cup \dots \cup A_n, \{g_1, \dots, g_n\})|/n + n\lambda$.
\end{lemma}

\begin{proof}
    Each bundle $A_i$ is assigned exactly one good in $\{g_1, \dots, g_n\}$. Therefore, we can consider an $n \times n$ matrix $Z = (z_{i, j})$ such that $z_{i,j} = d_{A_j}(g_i)$, and the violation increase can be written as:
    \begin{equation*}
        f_Z(\sigma) = z_{1, \sigma(1)} + \dots + z_{n, \sigma(n)}
    \end{equation*}
    where $\sigma \in S_n$ represents the assignment.
    
    If we choose a row and add $x$ to all its entries, then $f_Z(\sigma)$ increases by $x$ for all $\sigma \in S_n$, so $\max_{\sigma \in S_n} f_Z(\sigma) - \min_{\sigma \in S_n} f_Z(\sigma)$ remains unchanged. Therefore, for a new $n \times n$ matrix $Z' = (z'_{i,j})$ such that $z'_{i,j} = z_{i,j} - z_{i,1}$, we have:
    \begin{equation}
        \label{eq:profile_matrix1}
        \max_{\sigma \in S_n} f_Z(\sigma) - \min_{\sigma \in S_n} f_Z(\sigma) = \max_{\sigma \in S_n} f_{Z'}(\sigma) - \min_{\sigma \in S_n} f_{Z'}(\sigma)
    \end{equation}
    Since $z'_{i,j} = (p_{g_i})_{j-1} \ (j \geq 2)$ and $z'_{i,1} = 0$, the difference between the maximum and minimum values in each column is at most $\lambda$. Therefore, we have:
    \begin{multline}
        \label{eq:profile_matrix2}
        \max_{\sigma \in S_n} f_Z(\sigma) - \min_{\sigma \in S_n} f_Z(\sigma) = \max_{\sigma \in S_n} f_{Z'}(\sigma) - \min_{\sigma \in S_n} f_{Z'}(\sigma) \\
        \leq \left(\sum_{j=1}^n \max_{i \in \{1, \dots, n\}} z'_{i,j}\right) - \left(\sum_{j=1}^n \min_{i \in \{1, \dots, n\}} z'_{i,j}\right) = \sum_{j=1}^n \left(\max_{i \in \{1, \dots, n\}} z'_{i,j} - \min_{i \in \{1, \dots, n\}} z'_{i,j}\right) \leq n\lambda
    \end{multline}
    where the first equation is due to \autoref{eq:profile_matrix1}.
    
    Next, we consider the expected value of $f_Z(\sigma)$ when we choose $\sigma \in S_n$ uniformly at random. We calculate this in the following way:
    \begin{multline*}
        \mathbb{E}_{\sigma \in S_n}[f_Z(\sigma)] = \sum_{i=1}^n \mathbb{E}_{\sigma \in S_n}[z_{i,\sigma(i)}] = \frac{1}{n} \sum_{i=1}^n \sum_{j=1}^n z_{i,j} = \frac{1}{n} \sum_{i=1}^n \sum_{j=1}^n d_{A_j}(g_i) \\ = \frac{1}{n} \sum_{i=1}^n \sum_{j=1}^n |E(A_j, \{g_i\})| = \frac{1}{n} |E(A_1 \cup \dots \cup A_n, \{g_1, \dots, g_n\})|
    \end{multline*}
    The first equation is due to the linearity of expectations. The second equation is due to $\mathbb{E}_{\sigma \in S_n}[z_{i,\sigma(i)}] = \frac{1}{n} \sum_{j=1}^n z_{i,j}$ for each $i$. This holds because there are exactly $(n-1)!$ permutations such that $\sigma(i) = j$, for each $j \in \{1, \dots, n\}$.
    
    Thus, letting $s = |E(A_1 \cup \dots \cup A_n, \{g_1, \dots, g_n\})|$, we have $\min_{\sigma \in S_n} f_Z(\sigma) \leq s/n \leq \max_{\sigma \in S_n} f_Z(\sigma)$. Combining this with \autoref{eq:profile_matrix2}, we obtain $\max_{\sigma \in S_n} f_Z(\sigma) \leq s/n + n\lambda$, which proves \autoref{lem:matching}.
\end{proof}

\begin{figure}[htbp]
	\centering
	\includegraphics[width=0.6\textwidth]{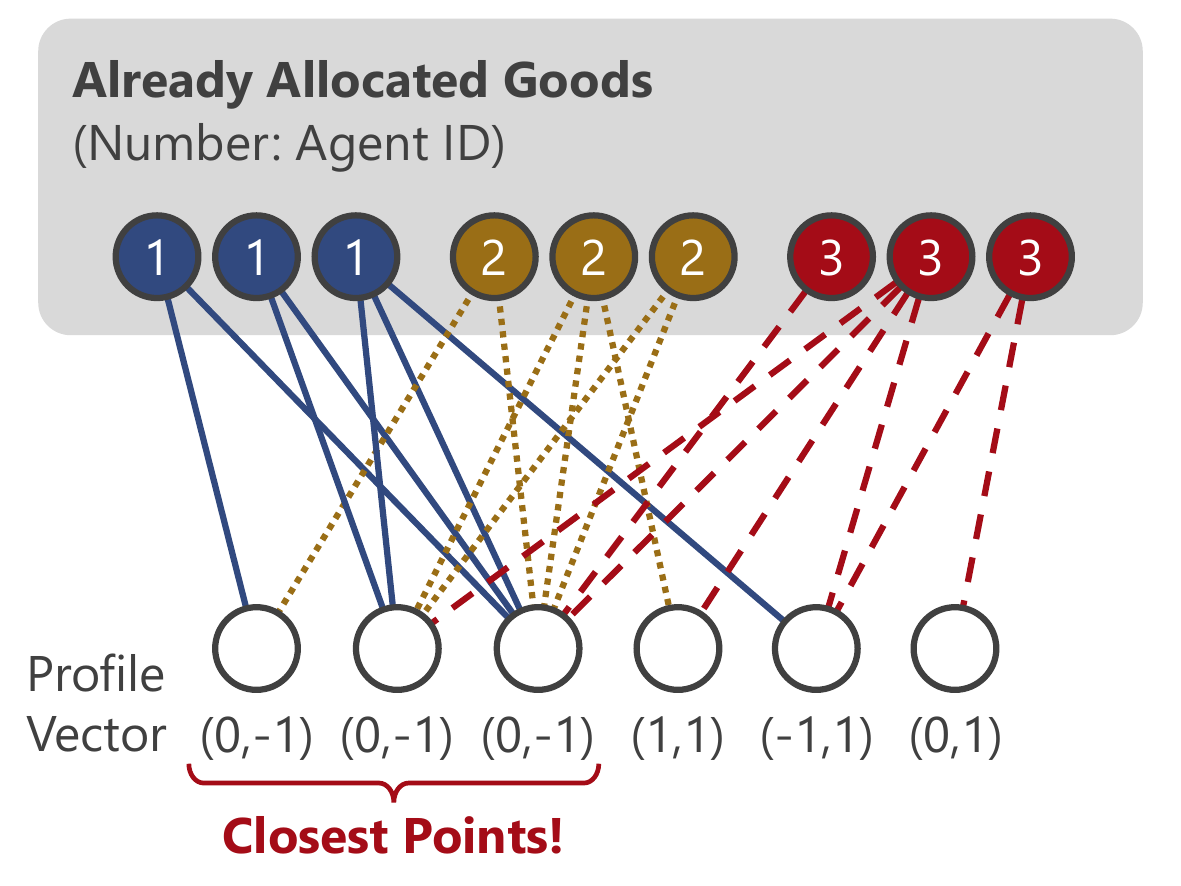}
	\caption{An example of choosing $n = 3$ goods. When we choose three goods with profile vector $(0, -1)$, the violation increase is the same regardless of the adversary's choice.}
	\label{fig:profile-vector}
\end{figure}

In fact, as long as there are sufficiently many remaining goods, choosing such $n$ goods is always possible.

\begin{lemma}
    \label{lem:hypercube}
    Let $d \in \mathbb{N}$ be a constant. Let $p_1, \dots, p_{kn}$ be points in a $d$-dimensional hypercube $[-\Delta, \Delta]^d$. Then, there exist $n$ points whose coordinates differ by $O(\Delta \cdot k^{-1/d})$ in each dimension.
\end{lemma}

\begin{proof}
    Let $q = \lfloor k^{1/d} \rfloor$. We divide $[-\Delta, \Delta]^d$ into $q^d$ smaller hypercubes with edge length $2\Delta/q = O(\Delta \cdot k^{-1/d})$. Since $q^d \leq k$, the pigeonhole principle shows that at least one hypercube contains $n$ or more points. We can choose any $n$ of these points to satisfy the condition of \autoref{lem:hypercube}.
\end{proof}

\subsection{The Algorithm}

To explain the algorithm, we assume that $m$ is divisible by $n$ (see ``dummy goods'' in \autoref{sec:prelim}). First, we consider the case where the graph has bounded degrees.

\begin{lemma}
    \label{lem:bounded-degree}
    When each vertex in $G$ has degree at most $\Delta$, there exists an EF1 allocation with at most $|E|/n + O(\Delta \cdot m^{1-1/(n-1)})$ violations, for any constant $n \geq 3$.
\end{lemma}

\begin{proof}
    The profile vectors are always inside the hypercube $[-\Delta, \Delta]^{n-1}$. At the $i$-th round of the game, there are $(m/n-i+1) \cdot n$ remaining goods, so due to \autoref{lem:hypercube}, there exist $n$ goods such that the coordinates of their profile vectors differ by only $O(\Delta \cdot (m/n-i+1)^{-1/(n-1)})$ in each dimension. Note that the fact that the profile vectors may change after each round does not affect the analysis.

    For $i = 1, \dots, m/n$, let $M_i$ be the set of goods chosen in the $i$-th round. Then, by \autoref{lem:matching}, the violation increase is upper-bounded by:
    \begin{equation*}
        \frac{|E(M_1 \cup \dots \cup M_{i-1}, M_i)|}{n} + O\left(\frac{n\Delta}{\left(m/n-i+1\right)^{1/(n-1)}}\right)
    \end{equation*}
    Summing over $i = 1, \dots, m/n$, the first term is bounded by $|E|/n$, and the second term is $O(\Delta \cdot m^{1-1/(n-1)})$, since $\frac{1}{1} + \dots + \frac{1}{k^{1/(n-1)}} = \Theta(k^{1-1/(n-1)})$ for $n \geq 3$ and $n$ is constant. Therefore, \autoref{lem:bounded-degree} holds. Note that the computed allocation is EF1 because it is based on the Biswas \& Barman's algorithm. The algorithm is presented as \textsc{DegreeEF1} (\autoref{alg:bounded_degree}).
\end{proof}

\begin{algorithm}
	\caption{$\textsc{DegreeEF1}(N, M, \mathcal{V}, G, \Delta)$}
	\label{alg:bounded_degree}
	\begin{algorithmic}[1]
        \State $\mathcal{A} \gets (\emptyset, \dots, \emptyset)$ \Comment{$\mathcal{A}$: current allocation}
        \State $R \gets M$ \Comment{$R:$ the set of remaining goods}
        \State $p_g \gets (0, \dots, 0)$ for each $g \in R$ \Comment{$p_g$: profile vectors}
        \For{$i = 1, \dots, m/n$}
            \State Choose $n$ points among $(p_g)_{g \in R}$ using \autoref{lem:hypercube}, and let $M_i$ be the set of chosen goods
            \State Assign goods in $M_i$ into $\mathcal{A}$ using the Biswas \& Barman's algorithm (lines 3--9 of \autoref{alg:biswas_barman}).
            \State $R \gets R \setminus M_i$
            \State Update the profile vectors $(p_g)_{g \in R}$
        \EndFor
		\State \textbf{return} $\mathcal{A}$
	\end{algorithmic}
\end{algorithm}

However, the degrees of vertices in $G$ may vary widely. Therefore, we use the strategy of grouping vertices by degree and running \textsc{DegreeEF1} starting from the highest-degree groups.

Technically, we divide the vertices into groups $L_0, \dots,$ $L_{\lceil \log_2 m \rceil}$: $L_0$ is the set of the top $n \lceil \sqrt{|E|} \rceil$ highest-degree vertices, $L_1$ is the set of the next $n \lceil \sqrt{|E|} \rceil$ highest-degree vertices, $\dots$, $L_i$ is the set of the next $2^{i-1} n \lceil \sqrt{|E|} \rceil$ highest-degree vertices, and so on. Note that later groups may be empty, but assuming $|E| \geq 1$, all goods are assigned to at least one group, due to $2^{\lceil \log_2 m \rceil} \geq m$. For the $|E| = 0$ case, every EF1 allocation has $0$ violations, so any algorithm to compute an EF1 allocation works.

\begin{itemize}
    \item For $i \geq 1$, there are $2^{i-1} n \lceil \sqrt{|E|} \rceil$ vertices with higher degree than the vertices in $L_i$. Since the sum of degrees in $G$ is $2|E|$, the degrees of vertices in $L_i$ are at most $2|E| / (2^{i-1} n \sqrt{|E|}) = \sqrt{|E|}/(2^{i-2} n)$. Therefore, we can run \textsc{DegreeEF1} to allocate the goods in $L_i$ with $\Delta = \sqrt{|E|}/(2^{i-2} n)$.
    \item For $i = 0$, we observe that in the $j$-th round of \textsc{DegreeEF1}, all profile vectors are inside $[-j, j]^{n-1}$. Therefore, we can run \textsc{DegreeEF1} to allocate goods in $L_0$ with $\Delta = \lceil \sqrt{|E|} \rceil$.
\end{itemize}

The algorithm is presented as \textsc{GraphEF1} (\autoref{alg:main}). Here, the subroutine \textsc{DegreeEF1} that starts from an allocation $\mathcal{A} = (A_1, \dots, A_n)$ and allocates a set of goods $T$ is denoted by $\textsc{DegreeEF1}(N, M, \mathcal{V}, G, \Delta, \mathcal{A}, S)$. In this case, we need to modify \autoref{alg:bounded_degree} to start from $\mathcal{A}$ (which deletes line 1) and initialize $R \gets S \setminus (A_1 \cup \dots \cup A_n)$ in line 2. We also need to compute the initial profile vectors in line 3.

\begin{algorithm}
	\caption{$\textsc{GraphEF1}(N, M, \mathcal{V}, G)$}
	\label{alg:main}
	\begin{algorithmic}[1]
        \State $\mathcal{A} \gets (\emptyset, \dots, \emptyset)$ \Comment{$\mathcal{A}$: current allocation}
        \State $R \gets M$ \Comment{$R:$ the set of remaining goods}
        \For{$i = 0, \dots, \lceil \log_2 m \rceil$}
            \State $L_i \gets$ the set of top $2^{\max(i-1, 0)} n \lceil \sqrt{|E|} \rceil$ highest-degree vertices among $R$
            \State $\Delta \gets$ $\sqrt{|E|}/(2^{i-2} n)$ if $i \geq 1$, and $\lceil \sqrt{|E|} \rceil$ if $i = 0$
            \State $\mathcal{A} \gets \textsc{DegreeEF1}(N, M, \mathcal{V}, G, \Delta, \mathcal{A}, L_i)$
            \State $R \gets R \setminus L_i$
        \EndFor
		\State \textbf{return} $\mathcal{A}$
	\end{algorithmic}
\end{algorithm}

\begin{theorem}
    \label{thm:general}
    There exists a balanced EF1 allocation with at most $|E|/n + O(|E|^{1-1/(2n-2)})$ violations, for any constant $n \geq 3$.
\end{theorem}

\begin{proof}
    We first prove the case where $|E| \geq 1$ and $m$ is divisible by $n$, in which \textsc{GraphEF1} can be used. First, the allocation is EF1 and balanced because it is based on the Biswas \& Barman's algorithm. Note that $|L_i|$ is divisible by $n$ for all $i = 0, \dots, \lceil \log_2 m \rceil$.

    Next, we analyze the number of violations. Let $E_i = E(L_0 \cup \dots \cup L_i)$. Due to \autoref{lem:bounded-degree}, when allocating the goods in $L_i \ (i \geq 1)$, the violation increase is at most the following (note that we ignore the factor $n$ in the $O(\cdot)$ notation, because $n$ is a constant):
    \begin{equation*}
        \frac{|E_i \setminus E_{i-1}|}{n} + O\left(\frac{\sqrt{|E|}}{2^i} \cdot \left(2^i \sqrt{|E|}\right)^{1-1/(n-1)}\right)
    \end{equation*}
    which is $|E_i \setminus E_{i-1}|/n + O(2^{-i/(n-1)}|E|^{1-1/(2n-2)})$. When allocating the goods in $L_0$, the violation increase is at most:
    \begin{equation*}
        \frac{|E_0|}{n} + O\left(\sqrt{|E|} \cdot \left(\sqrt{|E|}\right)^{1-1/(n-1)}\right)
    \end{equation*}
    which is $|E_0|/n + O(|E|^{1-1/(2n-2)})$. Summing these results over all $L_i$'s, the first term is $|E|/n$, and the second term is $O(|E|^{1-1/(2n-2)})$ because $\sum_{i=0}^\infty 2^{-i/(n-1)} = O(1)$. Therefore, the total number of violations is $|E|/n + O(|E|^{1-1/(2n-2)})$.

    When $|E| = 0$, we can create a balanced EF1 allocation using $\textsc{DegreeEF1}(N, M, \mathcal{V}, G, \Delta = 0)$, and it obviously has $0$ violations. When $m$ is not divisible by $n$, we create an EF1 allocation of $M \setminus M_{\mathrm{last}}$ by \textsc{GraphEF1}, where $M_{\mathrm{last}}$ is the set of $(m \bmod n)$ goods with lowest degrees. Then, we assign goods in $M_{\mathrm{last}}$, along with $n - |M_{\mathrm{last}}|$ dummy goods, by one iteration of the Biswas \& Barman's algorithm (lines 3--9 of \autoref{alg:biswas_barman}). Then, the resulting allocation is EF1 and balanced, as every agent receives zero or one dummy goods. Finally, we analyze the violation increase when assigning $M_{\mathrm{last}}$. When $m \leq n$, there are at most $\binom{m}{2} = O(1)$ edges, so the violation increase is $O(1)$. When $m > n$, the degree of vertices in $M_{\mathrm{last}}$ is at most $\frac{2|E|}{m-n}$, so the violation increase is at most $\frac{2|E|}{m-n} \cdot n = O(|E|/m) = O(|E|^{1-1/(2n-2)})$, due to $n \geq 3$ and $|E| = O(m^2)$. Therefore, the total number of violations is at most $|E|/n + O(|E|^{1-1/(2n-2)})$ in this case also.
\end{proof}

Note that the additive term of $O(|E|^{1-1/(2n-2)})$ matches the result of \autoref{lem:bounded-degree} when the graph is dense, i.e., $|E| = \Theta(m^2)$ and $\Delta = \Theta(m)$.

\subsection{Implementation}

We discuss an efficient implementation of the algorithm to eventually achieve linear time for any constant $n$.

\paragraph{Data structure for allocations.} We need to maintain the current allocation $\mathcal{A} = (A_1, \dots, A_n)$. In the Biswas \& Barman's algorithm, the bundles can be rearranged, which would take $O(m)$ time in a naive implementation. We resolve this by maintaining $n$ bundles $A'_1, \dots, A'_n$ together with a ``pointer'' from each agent to a bundle. When a rearrangement occurs, instead of explicitly changing the bundles $A'_1, \dots, A'_n$, we simply modify the pointers; this takes $O(1)$ time. We also maintain $v_i(A'_j)$ for all $(i, j)$; since valuations are additive, we can update these values in $O(1)$ time whenever a good is added. This provides sufficient information to compute the envy graph. Moreover, we change the definition of profile vectors as $p_g = (|N_{A'_2}(g)| - |N_{A'_1}(g)|, \dots, |N_{A'_n}(g)| - |N_{A'_1}(g)|)$ to avoid recomputation after each rearrangement. Then, in \textsc{DegreeEF1}, updates of the $p_g$'s occur only $O(|E|)$ times in total.

\paragraph{Data structure for ``closest points''.} In \autoref{lem:hypercube}, given $nk$ points in a $d$-dimensional hypercube, we divide it into $O(k)$ small regions to obtain $n$ points that are close to each other. However, recalculating the division whenever $q = \lfloor k^{1/d} \rfloor$ changes can be time-consuming. Instead, we set $q = 2^{\lfloor \log_2 (k^{1/d}) \rfloor}$. Since $k$ decreases by $1$ in each round, we update the division only when $k$ becomes less than $2^d, 2^{2d}, 2^{3d}, \dots$.

We need to create a data structure that supports the following queries: point insertion, point deletion, and finding $n$ points in the same region. For each region $r$, we maintain $S_r$, the multiset of points in this region. We also maintain the set $T = \{r: \text{region} \mid |S_r| \geq n\}$. We use hash tables to maintain $S_r$ and $T$, which support insertion, deletion, and element lookup in $O(1)$ time. Then, each query in the main data structure can be performed in $O(1)$ time. Note that, for point insertion/deletion, $T$ may be updated as the size of $S_r$ changes, but only one element changes in this case. For the initialization of the data structure, it only deals with the following number of points:
\begin{equation}
    k_0 + 2^{\lfloor (\log_2 k_0)/d \rfloor d} + \dots + 2^{2d} + 2^d = O(k_0)
\end{equation}
where $k_0$ is the initial number of points. Therefore, the overall time complexity for this data structure is $O(k_0 + (\text{number of queries}))$.

\begin{theorem}
    \textsc{GraphEF1} runs in $O(m + |E|)$ time for any constant $n \geq 3$.
\end{theorem}

\begin{proof}
    First, we sort the vertices by degree to construct $L_0, \dots, L_{\log_2 m}$. This takes $O(m)$ time using counting sort, since the degrees are at most $m$. Each round of the Biswas \& Barman's algorithm can be performed in $O(1)$ time, by the discussion of the data structure for allocations, and thus takes $O(m)$ time in total. For \textsc{DegreeEF1} to allocate the goods in $L_i$, the number of updates of profile vectors is $O(|E_i \setminus E_{i-1}|)$, where $E_i = E(L_0 \cup \dots \cup L_i)$. Therefore, by the discussion of the data structure for ``closest points'', \textsc{DegreeEF1} takes $O(|L_i| + |E_i \setminus E_{i-1}|)$ time. Summing over all $i$, we obtain a total running time of $O(m + |E|)$.
\end{proof}

Note that in \autoref{thm:general}, we considered the cases that $|E| = 0$ or $m$ is not divisible by $n$. Here, we are using \textsc{DegreeEF1} or the Biswas \& Barman's algorithm, so by the argument above, they only require $O(m + |E|)$ additional time.

\section{Conclusion}

We presented an algorithm for general additive valuations that computes an EF1 allocation with $(1+o(1))|E|/n$ violations. The $|E|/n$ bound has not yet been achieved, but since we are already within a $1+o(1)$ factor, we believe that the following conjecture is true:

\begin{conjecture}
    There exist EF1 allocations with at most $|E|/n$ violations.
\end{conjecture}

We can also consider a weighted version, in which each edge is assigned a weight indicating the importance of the conflict. For identical additive valuations, the cyclic shift round-robin algorithm also works in the weighted setting, so we can obtain an EF1 allocation such that the sum of the weights of conflicting edges is at most the baseline (a $1/n$ fraction of the total weight). For general additive valuations, \textsc{DegreeEF1} also works in the weighted setting, so we can obtain an EF1 allocation where the sum of the weights of conflicting edges is at most the baseline plus $O(\Delta \cdot m^{1-1/(n-1)})$, where $\Delta$ is the maximum value of the sum of weights of edges incident to a vertex. This algorithm works best when the weighted graph is near-regular.

Finally, we believe that our technique, which combines the Biswas \& Barman's algorithm with the ``closest point'' argument, can be used to obtain approximation results in fair division, especially under some entangled constraints. We give one example:

\begin{description}
    \item[\textbf{Problem.}] In the fair division of indivisible goods, each good $g \in M$ is assigned a score for each of $k$ features (price, beauty, etc.), $x_{g, 1}, \dots, x_{g, k} \in [0, 1]$, that is common to all agents. Note that the ``score'' is a completely different concept from valuations. Is it possible to achieve EF1 (on valuations) while keeping the difference in the sum of scores for each feature between agents small?
\end{description}

In this problem, we can set the profile vector as $p_g = (x_{g, 1}, \dots, x_{g, k}) \in [0, 1]^k$ and run \textsc{DegreeEF1} with $\Delta = 1$. Then, we obtain an EF1 allocation where the difference in the sum of scores for each feature is $O(m^{1-1/k})$ for any constant $k \geq 2$. For related work, Manurangsi \& Suksompong (2022) \cite{MS22} study fair division with $k$ valuation functions per agent and show the existence of an allocation that is envy-free up to $O(\sqrt{n})$ items with respect to all valuation functions. However, our result is different in that the allocation needs to be EF1 for a single valuation function.

\bibliographystyle{plain}
\bibliography{ref}

\end{document}